\documentclass[aps,pra,reprint,10pt,a4paper]{revtex4-1}
\usepackage[utf8]{inputenc}
\usepackage[sc,osf]{mathpazo}
\usepackage[T1]{fontenc}
\usepackage{amsfonts}
\usepackage{amssymb}
\usepackage{amsmath}
\usepackage{amsthm}
\usepackage{graphics}
\usepackage{graphicx}
\usepackage{braket}
\usepackage[colorlinks=true,linkcolor=blue,
urlcolor=blue,citecolor=blue]{hyperref}
\newtheorem{theorem}{Theorem}
\newtheorem{corollary}{Corollary}

\begin{document}

\title{
Unifying Two Notions of Nonlocality in Quantum Theory
}
\author{Shiladitya Mal, Aditi Sen(De)}

 \affiliation{Quantum Information and Computation Group,
Harish-Chandra Research Institute, HBNI, Chhatnag Road, Jhunsi, 
Allahabad 211 019, India}

\begin{abstract}

Ensembles containing orthogonal product states are found to be indistinguishable under local operations and classical communication (LOCC), thereby showing irreversibility in the preparation and distinguishing processes, which is commonly known as nonlocality without entanglement. On the other hand, correlations arising from incompatible measurements on entangled states lead to Bell-nonlocality. We unify these two concepts from the change  in certain property incurred in the ensemble under a suitable global unitary transformation. Specifically, we prove that under controlled-NOT (CNOT) operation, a full product basis can create entangled states if and only if the full bases or any subspace of it become irreducible in the process of LOCC discrimination. The proposed criteria quantifies the amount of nonlocality associated with the sets of product states which are even incomplete. For a set having entangled states,  we modify the quantity accordingly and  show that it can provide an explanation for the phenomena of more nonlocality with less entanglement.

\end{abstract}

\maketitle

\section{Introduction}
\label{sec_intro}

Nonlocality associated with quantum theory had been much debated and investigated from the time of its inception \cite{speakablebook}. Einstein, Podolski and Rosen (EPR) pointed out that certain type of quantum correlation \cite{HHHHRMP}, emerging from incompatible local measurements on composite quantum systems, violates the 
assumption of locality and reality \cite{epr} and hence it was argued that the quantum theory is not  complete  in the sense that it misses some variable to demonstrate a grand narrative of nature \cite{epr}. In response to the EPR argument, Bohm proposed an ontological model of quantum theory which restores reality but not locality \cite{bohm}. Subsequently, in 1964, Bell proved that any theory compatible with quantum theory cannot satisfy the above mentioned two assumptions simultaneously \cite{Bell, Bellexpt}. 
In spite of  Bell's comment, "correlations cry out  for explanation", the
second revolution of quantum technologies are heavily based on 
quantum nonlocality which leads to device-independent quantum information processing \cite{Barrett'05, Acin'06, Acin'07, Bellrev}.

On the other hand, in 1999, Bennett \emph{et al.}  provided  a set of product states, referred as
nonlocality without entanglement (NLWE) which cannot be distinguished by local operation and classical communication (LOCC) \cite{nlwe}, thereby adding more cause to the Bell's quotation.
In the context of distinguishability, it was shown that global operations can give more information  compared to  LOCC  \cite{PeresWootters}. On the other hand,  a set of orthogonal states can always be distinguished  by global measurements while nonorthogonal states cannot. Instead of global operations, if one allows only a restricted class of operations like LOCC or separable operations,  it is expected that  there will be  even a set of  orthogonal quantum states, containing especially entangled states  which cannot be discriminated by LOCC. Surprisingly, going against the intuition, exhibition of the new type of nonlocality does not require entanglement in their preparation nor in joint measurement which distinguishes them \cite{nlwe, ChenLi04}. Moreover, it was shown that any two orthogonal quantum states are always distinguishable by LOCC, independent of the content of entanglement \cite{WHSV}.

 In this way, two types of nonlocality,  studied independently, emerge in the literature. As pointed out by Niset and Cerf \cite{NisetCerf},  entanglement does not imply Bell-type nonlocality \cite{Werner} while a new type of nonlocality in terms of LOCC indistinguishability does not guarantee entanglement in shared states \cite{nlwe} -- we refer Bell-nonlocality as type-I while Bennett-nonlocality as type-II interchangeably. Moreover,  a toy model was proposed where several phenomena featuring quantum theory including NLWE was demonstrated although  the theory does not exhibit Bell nonlocality, thereby again separating the two concepts \cite{Spekkens'07}. Type-II nonlocality  initiates a plethora of  results after its introduction -- a set of unextendible product basis (UPB) \cite{UPB1, UPB2}  are found to be LOCC indistinguishable and also  provides a methodical construction of bound entangled states,  several full and incomplete LOCC indistinguishable product ensembles are proven \cite{groisman, Rinaldis, Nathanson, others}. On the other hand, it was also shown that there exist sets of entangled states or sets of entangled as well as product states that cannot be distinguished by LOCC \cite{WHSV, Ghosh'01, WalgateHardy, morenlwe, Somdamnlp, maxentdxd}. 
 Another interesting twist came into the picture with the concept of more nonlocality with less entanglement \cite{morenlwe} --  it was shown that although three maximally entangled states can be distinguished by LOCC while when an entangled state is replaced by a product state, thereby reducing the average entanglement in the ensemble, the ensemble becomes LOCC-indistinguishable. Phenomena of local-indistinguishability find applications in the area like  quantum secret sharing and data hiding technology which are primitives of cryptographic protocols \cite{datahide}, the capacities of classical information transfer via quantum states \cite{DCamader, DCnoise}. 
 
All these counter-intuitive results make characterizing  LOCC-indistinguishable sets of states to be highly non-trivial. Although there are several works in  which  specific ensembles are extensively studied  from the perspective of distinguishability by LOCC, investigations towards  the quantification of such ensembles are limited \cite{locallyaccess1, locallyaccess,  measureQCset, locallyaccesslow}. For this purpose, like Holevo bound \cite{Holevo} for accessible information with global measurements, a universal upper bound, known as local Holevo bound, \cite{locallyaccess1, locallyaccess} based on aquiring  information via LOCC was found \cite{locallyaccess1, locallyaccess}. Interestingly, for the locally indistinguishable product ensembles, the local Holevo bound cannot provide a better bound than the Holevo bound with global operations. In  a similar spirit of the lower bound obtained in Ref. \cite{Josza94} for global actions, the lower bound on locally accessible information has also been obtained when one of the parties is two dimensional \cite{locallyaccesslow}. 

In the present work,  we connect these two notions of nonlocality which were inequivalent according to the conventional wisdom, thereby leading to a characterization of LOCC-indistinguishable sets. 
In particular, these sets can be classified into different categories --  (1) sets  that can be distinguished by one round of LOCC, independent of the parties who start the protocol, (2) an ensemble which requires a specific party to start with, (3) sets of states which require two or more rounds of LOCC and finally (4) sets of states which are indistinguishable via infinite rounds of LOCC. Even among the distinguishable  states, sets mentioned in  (1), (2) and (3) surely possess different  properties which are  not  well understood. Here we investigate the effect of some specific transformations on the whole ensemble and find that it has remarkable implications in characterizing these sets.
Considering full product bases, we show that LOCC-indistinguishable sets always generate entanglement under such transformation which is taken to be a global entangling operation, i.e., controlled-NOT (CNOT) operation, having potential to demarcate the above four categories as well. Hence, we reveal a significant property that these indistinguishable sets of product states have "potential to be entangled". Since the generated state is also pure, its entanglement content can be uniquely  quantified by the von-Neumann entropy of the reduced subsystem \cite{Bennett96} and we call it as the "nonlocal entropy" associated with the initial product bases.

Specifically, we prove that for the full product basis of arbitrary dimensional bipartite system, nonlocal entropy is positive if and only if the bipartite product states  either form irreducible basis or become irreducible \cite{Rinaldis} in any rounds of the LOCC protocol. Moreover,  we show that pure entangled states  generated from LOCC-indistinguishable sets of product states always exhibit Bell-nonlocality \cite{gisin'91} while LOCC-distinguishable sets may not. From this idea, we can provide a general method for quantifying nonlocality associated with an arbitrary set of product states with respect to a given protocol to probe local distinguishability.
 The quantification can also be interpreted from a thermodynamic perspective, i.e., by linking  nonlocal entropy with a reduction in the amount of work which can be extracted from  subsystems locally \cite{wd}. 
 We then generalize this idea to incomplete product bases and sets containing entangled states \cite{UPB1, UPB2, WalgateHardy, morenlwe, Somdamnlp, maxentdxd}.
If the set of states contains atleast one entangled state, we propose, reduction of local entropy from the average ensemble state after application of CNOT, as the pertinent figure of merit which can successfully quantify the LOCC-indistinguishability. Most significantly, this novel treatment uncovers a physical insight for the phenomena of more nonlocality with less entanglement \cite{morenlwe}. Finally, we show how to obtain upper and lower bounds on locally accessible information  based on the proposed quantifier.

\section{Characterising irreducible product bases via nonlocal entropy}
\label{sec:avent}

 In this section, we characterize local distinguishability of sets consisting of only product states -- full product bases (FPB) as well as  incomplete product bases and show how Bell-nonlocality emerges from Bennett-nonlocality. Incomplete product bases are called completable if it can be made complete only by product states  and are uncompletable  if its complementary subspace contains fewer  number of product states than its dimension. UPB is a prominent example of the last one.   
Let us first define a physical quantity, dubbed here as nonlocal entropy, which not only connects Bell-type nonlocality to Bennett-type nonlocality but also captures the amount of nonlocality present in the product orthogonal ensembles, thereby quantifies it.  Moreover, we will show that the  measure introduced here has potential to disclose the hierarchies in the  degree of nonlocality present in these ensembles. 

To set the stage, suppose the ensemble, \(\{p_i, |\psi_i\rangle^{AB} = |\psi_i^A\rangle \otimes |\psi_i^B\rangle\}\}_{i=1}^{k}\), consisting of \(k\) states which are produced locally by two distantly situated parties, \(A\) and \(B\) with probability $p_i$. 
Since the states are orthogonal,  they can always be distinguished  by global operations. Moreover, individual  states in the ensemble are product, and hence they do not  show any nonclassical correlation, useful for quantum information processing tasks \cite{HHHHRMP}. However, it is highly nontrivial that there exists ensembles of product states which are locally indistinguishable \cite{nlwe, UPB1, UPB2}. 

Let us now apply a transformation which is taken here as a global unitary on the set. We define nonlocal entropy in the following way:
\begin{eqnarray}
 \delta_S^{\rightarrow} & = & \max_{U_{AB}} [\sum_i p_i  \min_{U_A} (S(\rho^{'A}_i) - S(\rho^A_i))] \nonumber\\
&\equiv &  \max_{U_{AB}} \sum_i p_i  \min_{U_A} E(|\psi'_i\rangle^{AB});\\
 &=&  \max_{U_{AB}} [ \sum_i p_i \min_{U_A} ( W_{i_{A}}^{in} - W_{i_{A}}^{fin})],
\label{eq_prodright}
\end{eqnarray}
where \(S(\sigma) = - \mbox{tr} \, \sigma \log_2 \sigma \) denotes the von-Neumann entropy which vanishes for pure states, \(\rho_i^A\) is the local density matrix of the individual states in the ensemble and \(\rho^{'A}_i(=tr_B[ U_{AB}|\psi^{AB}_i\rangle \langle \psi^{AB}_i| U_{AB}^\dagger ]\)) is that of transformed ensemble. $U_{AB}, U_A$ are global and local unitaries respectively and \(E\) is an entanglement measure \cite{Bennett96}. 
For pure states, since \(S(\rho^{'A})\) is also the entanglement of the output state \cite{Bennett96}, we can interpret  \(\delta_S^{\rightarrow}\) as the capability of entanglement production of a given ensemble. Therefore, it captures the potential of the initial ensemble to be entangled and we call it as nonlocal entropy to differentiate from local entropy of the initial ensemble, which is always zero for product states. 

Moreover, \(W_{i_{A}}^{in} = \log_2 d^A - S(\rho_i^A) \)  is the maximal possible extractable work from side \(A\)  from  each states of the ensemble with \(d^A\) being the dimension of \(A\) while \( W_{i_{A}}^{fin} = \log_2 d^A - S(\rho^{'A}_i) \) represents the same for resulting ensembles.
Hence, nonlocal entropy can also be interpreted as the maximal  reduction of  work extractable in local subsystems of each state in the ensemble,
due to the application of global unitary operations.
The superscript arrow (\(\rightarrow\)) represents the directionality of the transformation. Specifically, when CNOT is taken as a global transformation, \(\rightarrow\) indicates that \(A\) acts as control and \(B\) as target which we will then show to be connected  with the irreducibility of sets.
Similarly, we can have \( \delta_S^{\leftarrow} \), in which role of  \(A\) and \(B\) gets reversed and 
  consequently,  we define a symmetric  nonlocal entropy as
\begin{eqnarray}
\delta_S = \frac{1}{2}(\delta_S^{\rightarrow}  + \delta_S^{\leftarrow}).
\label{eq_prodsymm}
\end{eqnarray}

In next subsection we begin with some examples to show how nonlocal entropy can characterize different sets of FPB to motivate our main result described subsequently.

\subsection{Illustration with two qubits} 

In this subsection, we consider 
product ensembles of two qubits. We know that for two-qubits, there is no example of nonlocality without entanglement \cite{nlwe} as well as UPB \cite{UPB1, UPB2}. We, however, show that the hierarchies present even in two-qubit product states can be well demonstrated by nonlocal entropy.

\emph{Case 1.} Let us first consider the ensemble \(\mathcal{E}^1 = \{|00\rangle, |01\rangle, |10\rangle, |11\rangle \}\). Without loss of generality, we assume that the states are given with equal probability.   A possible example of a global unitary operation can be the controlled-NOT (CNOT) operation. Moreover, we know \cite{Barenco95, Zanardi00, FV04} (cf. \cite{Chitambar}) that any two-qubit unitary operations can be realized by a few single-qubit unitary operators and three CNOT gates upto a global phase. Since entanglement remains invariant under local unitary  operations,  CNOT gates are enough for computing \(\delta_S\). 
After CNOT operation, the   ensemble remains invariant and is independent of the sites which act as  control and target.  Hence, \(\delta_S\) vanishes for this ensemble. 

\emph{Case 2.}  We now consider another ensemble  \(\mathcal{E}^2 = \{|0+\rangle, |0-\rangle, |10\rangle, |11\rangle \}\), where \(|\pm \rangle = \frac{1}{\sqrt{2}}( |0\rangle + |1\rangle )\). To distinguish states of this ensemble by means of LOCC, \(A\) should start the protocol \cite{groisman}.  If \(B\) starts the protocol, \(\mathcal{E}^2\) cannot be distinguished by LOCC deterministically. If \(A\) acts as a control qubit,  it is easy to check that \(\mathcal{E}^2\) remains invariant under CNOT operation and hence \(\delta_S^{\rightarrow}\) vanishes. On the other hand, if CNOT operation is performed with \(B\) being a control qubit, \(\mathcal{E}^2 \rightarrow \{\frac{1}{\sqrt{2}}(|00\rangle \pm |11\rangle), |10\rangle, |01\rangle \}\), and hence \(\delta_S^{\leftarrow} = \frac{1}{2}\) which leads to \(\delta_S = 0.25\).   

Let us now move to a general product bases consisting of four arbitrary product states \cite{WalgateHardy}  in \(\mathcal{C}^2 \otimes \mathcal{C}^2\). The proof will be based on the notion of irreducibility of  the  product basis \cite{Rinaldis} which is connected to lcoal distinguishability of FPB.  In \(\mathcal{C}^{d_1} \otimes \mathcal{C}^{d_2}\), an FPB is said to be irreducible if it cannot be divided into two blocks of states which lie in  \(\mathcal{C'}^{d_1} \otimes \mathcal{C}^{d_2}\) and in the complementary subspace, \(\mathcal{C'}^{d_1}_{\perp} \otimes \mathcal{C}^{d_2}\) or into \(\mathcal{C}^{d_1} \otimes \mathcal{C'}^{d_2}\bigoplus\mathcal{C}^{d_1} \otimes \mathcal{C'}^{d_2}_{\perp}\).

\begin{theorem}
Nonlocal entropy of any two-qubit full product basis is non-zero if and only if the full product basis is irreducible from the side who starts the protocol to distinguish.
\end{theorem}

\begin{proof}
Let us first consider \emph{if} part.  In \(\mathcal{C}^2 \otimes \mathcal{C}^2\), an ensemble having  four orthogonal states which are locally distinguishable when \(A\) starts the protocol, takes the form  \cite{WalgateHardy} as
\begin{eqnarray}
\label{eq_4product}
\mathcal{E}^4 &=& \{|\psi_1\rangle = |0 \eta_1\rangle; \, |\psi_2\rangle =|1 \eta_2\rangle; \nonumber\\
 && |\psi_3 \rangle = |0 \eta_1^{\perp}\rangle; \,  |\psi_4 \rangle = |1 \eta_2^{\perp}\rangle\},
\end{eqnarray}
where \(|\eta_i \rangle = a_i |0\rangle + b_i |1\rangle, \, i=1, 2\) and  \(|\eta_i^{\perp} \rangle = - \overline{b}_i |0\rangle + \overline{a}_i |1\rangle, \, i=1, 2\) and we assume that one of the \(a_i, b_i \neq 0\). 
In this case, when \(A\) starts the protocol, the ensemble is reducible and one can easily check that after CNOT operation from \(A\) to \(B\), entanglement production vanishes. On the other hand,  if \(B\) starts the protocol, the ensemble is irreducible provided one of the \(\{a_i, b_i\} \neq 0\). Note here that if \(b_i\)s are zero, the ensemble reduces to Case 1 discussed above, while \(a_1= b_1 = 1/\sqrt{2}, a_2=1\) gives the second case.  Let us now apply CNOT with \(B\) being the control qubit.   In this situation,  each  state in \(\mathcal{E}^4\) transforms to nonmaximally entangled state after CNOT operations and hence \(\delta_S^{\leftarrow} >0 \) when one of the \(\{a_i, b_i\}\)s are nonvanishing.

Let us now move to the \emph{only if-part}. Suppose entanglement is generated through CNOT operation on a two-qubit product state. Any product state is of the form given by, \(|\phi_1\rangle \otimes |\phi_2\rangle \equiv(a_1 |0\rangle + b_1 |1\rangle) (a_2 |0\rangle + b_2 |1\rangle), \mbox{with all}\, a_i, b_i \,\mbox{s are unequal}\). Suppose \(B\) acts as a control qubit. After CNOT,  the state becomes \(a_1 (a_2 |00\rangle + b_2|11\rangle) + b_1 (a_2 |10\rangle + b_2 |01\rangle\) which is entangled when \(a_2 \) and \(b_2\) do not vanish or  all $a_i, b_i$s are not equal. Let us consider that \(b_1=0\) and  $a_1=1$ which implies that the state reduces to \(|\psi_1\rangle\). Now we have to show that a FPB can be constructed including $|\psi_1\rangle$, which will be irreducible from \(B\)'s side. The orthogonal product vector  of \(|\psi_1\rangle\) can be either \(|\psi_2\rangle\) or \(|\psi_3\rangle\) in Eq. (\ref{eq_4product}). And if we fix to one of them, the other possibility leads to the ensemble, \(\mathcal{E}^4\), given in Eq. (\ref{eq_4product}). This set is irreducible if \(B\) starts the protocol. If none of the coefficients vanish, we can also have a FPB which after measurement by \(B\) leads to a state \(\{|\psi_1 \psi_2\rangle, |\psi_1^{\perp} \psi_2\rangle \}\) or  \(\{|\psi_1 \psi_2^{\perp}\rangle, |\psi_1^{\perp} \psi_2^{\perp}\rangle \}\) which is again irreducible if \(B\) starts the protocol  and hence the proof. 

\end{proof} 

\emph{Remark 1.} Since two-qubit  LOCC-distinguishable FPB are known \cite{WalgateHardy}, we note that they are distinguishable by  a single round of LOCC protocol.
In higher dimension, we know that there exists FPB which can be distinguished by several rounds of LOCC protocol \cite{UPB1}. Therefore, we have to modify the proof of the  theorem in higher dimension, which leads to a most general way of quantifying nonlocality of LOCC-indistinguishable set of product states, as we will show in succeeding section. 

\emph{Remark 2.}  Instead of full product basis, if  we consider the incomplete basis, the sufficient condition holds, i.e., if the states are indistinguishable even from one side, the entanglement generation is  possible as we will illustrate in higher dimensions by considering an example of UPB. 


\subsection{Higher dimensional Product states} 
\label{sec_highdimproduct}

Let us prove the necessary and sufficient condition for  FPB having arbitrary dimensional bipartite states. 

\begin{theorem}
Entanglement will be generated from the application of CNOT on a set of full product bases if and only if the full product bases contain an irreducible subspace  from the party who acts as a control. 
\end{theorem}

\begin{proof}
Suppose the LOCC protocol is started by one of the parties, say, \(A\). If she finds an irreducible subspace in FPB,  it means that \(A\) cannot design any further measurements which can divide the rest of the states into two blocks. It then immediately implies that there exists states in that subspace which is nonorthogonal at \(A\)'s side. Therefore, without loss of generality, we assume that the ensemble consists of states of the form \(\{|\psi_1 \rangle = |\eta_1\rangle |0\rangle; |\psi_2 \rangle = |\eta_2\rangle |1\rangle\}\) which is orthogonal from \(B\)'s part and where \(|\eta_i\rangle = \sum _{j=0}^{d_1-1} a_{ij} |i\rangle\).   In \(\mathcal{C}^{d_1} \otimes \mathcal{C}^{d_2}\),  the generalised CNOT operation can be defined as 
\begin{eqnarray}
|i,i\rangle\rightarrow |i,i\oplus_{d_2} i\rangle,
\end{eqnarray}
where \(\oplus_{d_2}\) represents the addition modulo \(d_2\). Applying the CNOT gate on  \(|\psi_i\rangle\)s where \(A\) is the control qubit, we immediately find that entanglement of the output state is nonvanishing.

Let us now concentrate on \emph{only if-part}.  We first notice that  in \(\mathcal{C}^{d_1} \otimes \mathcal{C}^{d_2}\),  full product basis should contain \(d_1 d_2\) number of  orthogonal product states and hence there will be  \(d_1\) states which are orthogonal from \(A\)'s side while the rests are orthogonal from \(B\)'s side. Now similar to Theorem 1, any product state in two-qudits is of the form \(\sum_{i=0}^{d_1-1} a_i|i\rangle \otimes \sum_{i=0}^{d_2-1} b_i|i\rangle\) with all \(a_i\) and \(b_i\)s being unequal.  It can generate entanglement after CNOT operations with \(A\) as control qubit   when all \(b_i\)s except one vanishes, thereby reduces to a state \(|\psi_1\rangle\), where we assume that \(b_1\neq 0\) or some of the \(b_i\) s  vanish or none of them vanish.  In all these situations, from the construction of FPB mentioned above, we can conclude that there exists a round of LOCC protocol where either the state at \(A\)'s side are all same or are nonorthogonal. In both these situations, there will be an irreducible subspace in the constructed FPB from \(A\)'s side and hence the proof. 
\end{proof}

\begin{corollary} If a set of product states is indistinguishable by LOCC, it is surely irreducible \cite{Rinaldis} and hence for locally indistinguishable FPB, entanglement production by CNOT operations always occurs.
\end{corollary} 

To illustrate  Theorem 2, let us now move on to the examples of product states which are indistinguishable with local operations and more than one round or infinite round of classical communication. We have chosen two examples -- one in   \(\mathcal{C}^{3} \otimes \mathcal{C}^{2}\) and another one in  \(\mathcal{C}^{3} \otimes \mathcal{C}^{3}\).  

\emph{Case 1.} In \(\mathcal{C}^2\otimes \mathcal{C}^d\), all FPB are LOCC-distinguishable. Let us consider a basis in \(\mathcal{C}^{3} \otimes \mathcal{C}^{2}\), consists of \(\{\frac{1}{\sqrt{2}}(|1\rangle \pm |2\rangle) |0\rangle, |11\rangle, |21\rangle, |00\rangle, |01\rangle\}\).  The states are distinguishable when \(B\) starts the protocol and also reducible. However, if \(A\) starts the protocol by projecting the states in \(|0\rangle\) and the subspace spanned by \(|1\rangle\) and \(|2\rangle\), then in the latter occasion, the resulting ensemble of first four states becomes irreducible. Therefore, \(A\) can not complete the task of distinguishing perfectly unless \(B\) helps and it is straightforward to see that average entanglement generated from the ensemble, if Alice acts as a control, is \(\delta_S^{\rightarrow} = 1/3>0\).  

\emph{Case 2: Nonlocality without entanglement \cite{nlwe}.} We now consider the famous example of full basis in \(\mathcal{C}^{3} \otimes \mathcal{C}^{3}\) , given by
\begin{eqnarray}
\label{eq_nlwe}
&&|\psi_1\rangle = |1\rangle |1\rangle; |\psi_2\rangle = |0\rangle |0+1\rangle;\nonumber\\
&&|\psi_3\rangle = |0\rangle |0-1\rangle; |\psi_4\rangle = |2\rangle |1+2\rangle;\nonumber\\
&& |\psi_5\rangle = |2\rangle |1-2\rangle; |\psi_6\rangle = |1+2\rangle |0\rangle\nonumber\\
&&|\psi_7\rangle = |1-2\rangle |0\rangle; |\psi_8\rangle = |0+1\rangle |2\rangle\nonumber\\
&& |\psi_9\rangle = |0-1\rangle |2\rangle,
\end{eqnarray} 
where \(|i+j\rangle = \frac{1}{\sqrt{2}}(|i\rangle + |j\rangle)\).  In case of product states, the second term in Eq. (\ref{eq_prodright}) always vanishes and nonavnishing value of the first term guarantees the states to be LOCC-indistinguishable. In this case, after CNOT  gate, we find that the last four states only contribute and therefore, we have
\begin{eqnarray}
\delta_S^{\rightarrow} =4/9 . 
\end{eqnarray}
It can be easily checked that the similar value can be obtained for \(\delta_S^{\leftarrow}\) and therefore,  nonlocal entropy is strictly positive.

Upto now, we considered the ensembles which are the full basis.  Let us discuss how the idea of characterising local distinguishability via entanglement creation through CNOT can be applied to the set of incomplete product basis. 
\begin{corollary}
For incomplete product basis, if the set can be  completed only by the product states and the full product basis is reducible in any round of the LOCC protocol, the entanglement generation is not possible. 
\end{corollary}
\begin{proof}
In Theorem 2, we prove that if the FPB is reducible via LOCC,  the corresponding ensemble cannot create entanglement. If the FPB is reducible, surely the subset is also reducible and hence cannot create entanglement. 
\end{proof}

Let us now consider the uncompletable product states which are indistinguishable by LOCC and hence they are irreducible. The above Corollary cannot shed light on them. \\
\emph{Case 3: Unextendible product bases \cite{UPB1, UPB2}. } If we relax the FPB to a set of orthogonal product states forming incomplete set, the nonlocal entropy may increase even  for the set which  is LOCC distinguishable. For example, the ensemble \(\{|00\rangle, |10\rangle, |+1\rangle\}\), can produce positive entanglement entropy. 
If the set is not extendible i.e., UPB,  the nonlocal entropy can be shown to be always positive. 
Let us consider the incomplete product  ensemble in two qutrits,  given by
\begin{eqnarray}
|\psi_1\rangle &= & \frac{1}{\sqrt{2}} |0\rangle ( |0\rangle - |1\rangle);  \, \, |\psi_2\rangle = \frac{1}{\sqrt{2}}  |2\rangle ( |1\rangle - |2\rangle) \nonumber \\
|\psi_3\rangle &= &\frac{1}{\sqrt{2}} (|0\rangle - |1\rangle) |2\rangle; \, \, |\psi_4\rangle =\frac{1}{\sqrt{2}} (|1\rangle - |2\rangle) |0\rangle; \nonumber \\
|\psi_5\rangle &=&\frac{1}{3} (|0\rangle + |1\rangle + |2\rangle)( (|0\rangle + |1\rangle + |2\rangle).
\label{eq_UPB}
\end{eqnarray}
After CNOT operations, we find \(\delta_S^{\rightarrow} = \delta_S^{\leftarrow} =(2+\log_2 3)/5 \).

We find that all these sets can produce entanglement via CNOT operations, thereby showing nonlocal nature of the ensembles. In this respect, let us state the following theorem. 

\begin{theorem}
For any locally indistinguishable product basis, nonlocal entropy is strictly positive.
\end{theorem}

\begin{proof}
If a product orthogonal ensemble cannot be distinguished locally, there always exist atleast a pair of two or more states in the ensemble whose one of the subsystems are nonorthogonal with other pair, i.e. if \(\{|\psi^A_i\rangle |\psi^B_i\rangle\}_{i=1}^{k}\) is locally indistinguishable, then it should consist of two states which without loss of generality, can be taken of the form 
\begin{eqnarray}
|\psi_1 \rangle = |\eta_1\rangle |\phi\rangle; |\psi_2 \rangle = |\eta_2\rangle |\phi^{\perp}\rangle,
\end{eqnarray} 
where  \(\eta_1 = \sum_{i=0}^{d_1-1} a_i |i\rangle\) and \( |\eta_2\rangle = \sum_{i=0}^{d_1-1} b_i|i\rangle\) and say, \(|\phi\rangle = |d-1\rangle \) and \( |\phi^{\perp} = |0\rangle\) are two orthogonal states. Here, we assume that during CNOT operation, \(A\) acts as a control. If we operate CNOT on these states, the resulting state can be written as
\begin{eqnarray}
|\psi_i \rangle \rightarrow \sum_{i=0}^{d_1-1} a_i |i, i\oplus_{d_2} i\rangle; |\psi_j \rangle\rightarrow \sum_{i=1}^{d_1-1} b_i |i,i \rangle,
\end{eqnarray} 
whose local density matrices have nonvanishing entropy, thereby creating \(\delta_S^{\rightarrow}>0\). The locally indistinguishable ensemble also consists of states where the \(A\)-part interchenages with  part-\(B\), and therefore, it also  gives    \(\delta_S^{\leftarrow}>0\), thereby showing \(\delta_S\) to be nonvanishing for any locally indistinguishable ensembles in \(\mathcal{C}^{d_1} \otimes \mathcal{C}^{d_2}\). Above characetristics is true both for complete as well as incomplete basis and hence it is true both for NLWE basis as well as UPB. Notice, moreover, that with the increase of dimension, number of such states increases and hence the amount of \(\delta_S\) also increases (eg.  see \cite{ChenLi04}). 
\end{proof}

We end this section by making a comment on how behaviour of ensemble under certain transformation enables quantification of nonlocality associated with the set corresponding to  a protocol that distinguishes them. Since, in general, a LOCC protocol includes multiple rounds of classical communication by both the parties, we propose a weighted nonlocal entropy as the proper quantifier. If at some round,  states which are not eliminated by previous measurements, become irreducible from both sides, only then we put non-zero weight to the amount of entanglement generated from those irreducible states via CNOT operation, otherwise we put zero-weight to it.

As our treatment of merging two kinds of nonlocality provides a way  of quantification, let us investigate the transformation  which  enables us to quantify ensembles containing also entangled state as its member.

\section{Identifying hierarchy  in sets having entangled state(s)}

We will now show that ensembles, which contain atleast one entangled state, behave differently than the set of product states when considered under certain global transformations. Towards quantifying nonlocality associated with such a set, we modify the methodology accordingly, which can be thought of a  "reverse method" than the one used for product states. It was already known that LOCC-indistinguishability is not direct connected with the entanglement content of the set of states, as shown by the examples of  local distinguihsability of two orthogonal states  \cite{WHSV} as well as by  "more nonlocality with less entanglement" \cite{morenlwe}.

Apart from characterising the sets, the proposed figure of merit provides an explanation for the phenomena of "more nonlocality with less entanglement" \cite{morenlwe}, which was missing in the literature. To the best of our knowledge, till now, there exists only a single natural quantifier of such ensemble of states which invokes the gap between accessible information extractable by local and global means \cite{measureQCset}.  

 Before presenting the modified figure of merit, let us illustrate the effect of CNOT on ensembles of two-qubit maximally entangled basis, \(\{ |\phi^{\pm}\rangle =\frac{1}{\sqrt{2}} (|0 0\rangle \pm  |11\rangle);  |\psi^{\pm}\rangle = \frac{1}{\sqrt{2}}( |0 1\rangle \pm  |10 \rangle)\}\). If CNOT acts from \(A\) to \(B\),  all the states becomes product, given by \(\{\frac{1}{\sqrt{2}}  (| 0\rangle \pm  |1\rangle)|0\rangle;   \frac{1}{\sqrt{2}}( |0 \rangle \pm  |1\rangle)|1\rangle\}\). The fact of disentangling entangled state through entangling operation is also noticed very recently and employed to define absolute entangled set \cite{aentset}.  
 
 Motivated by the action of CNOT on entangled states, we come up with the following quantity:  
For a given ensemble to be distinguished, \(\{p_i, |\psi_i^{AB}\rangle\}_{i=1}^{k}\), we consider the gap between the local entropy of the average initial state and that of transformed average ensemble by means of \emph{global} unitary operations. In general, it can be written as
\begin{eqnarray}
\Delta_S^{\rightarrow} = \max [\max_{U} \big( S(\overline{\rho}^A) - S(\overline{\rho'}^A), S(\overline{\rho}^B) - S(\overline{\rho'}^B)\big) ],
\label{eq_entright}
\end{eqnarray}
where \(\overline{\rho}^i,\, i= A, B \) is the reduced density matrix of the average initial state \(\overline{\rho}^{AB} = \sum_i p_i  |\psi_i^{AB}\rangle \langle \psi_i^{AB}|\) while  \(\overline{\rho'}^{AB} = \sum_i p_i  U|\psi_i^{AB}\rangle \langle \psi_i^{AB}|U^\dagger\), with \(\overline{\rho'}^i,\, i= A, B \)  being the local subsystems. One has to first maximize over all global unitaries, so that the each output state becomes as less entangled as possible. If all the states become product,  \(\overline{\rho'}^{AB}\) is separable and finally maximum has to be taken between two quantities obtained from subsystems. "\(\rightarrow\)" in the superscript indicates that if the unitary is CNOT, \(A\) acts as a control qubit while \(B\) is the target. 

If the states are locally distinguishable, the reduction in the amount of local entropy of the ensemble state can be made maximum via global unitary operations, i.e., it will be possible to have  \(\max[S(\overline{\rho}^A), S(\overline{\rho}^B)]\) where \(S(\overline{\rho'}^A)\) or \(S(\overline{\rho'}^B) \) is vanishing. 
On the other hand, for locally indistinguishable set of states, such a maximum local entropy reduction does not occur with the application of CNOT. It is to be noticed that unlike in the case of pure product bases, here we consider the reduced density matrix of the average state  of the ensemble which is a mixed state. Therefore, the entropy of  the local density matrix is not a faithful measure of entanglement. However,  the quantifier can still be interpreted from a thermodynamic perspective i.e.,  the maximal amount of  work which can be extracted by local parties from the average state and is given by
\begin{eqnarray}
\Delta_S^{\rightarrow} = \max [\max_{U} \big(  \overline{W}^A_{fin} - \overline{W}^A_{in}, \overline{W}^B_{fin} - \overline{W}^B_{in}  \big)],
\label{eq_entWork}
\end{eqnarray}
where \(\overline{W}_{in} = \log_2 d^i -  S(\overline{\rho}^i)\) and \(\overline{W}_{fin} = \log_2 d^i -  S(\overline{\rho'}^i)\) with \(i = A, B\). 

Similarly, one can define \(\Delta_S^{\leftarrow}\) when \(B\) is the control. We expect to have the same value for the ensembles whose  LOCC (in)distinguishability  do not depend on the party who starts the protocol. 
To obtain a symmetric quantity, we again introduce 
\begin{eqnarray}
\Delta_S = \frac{1}{2}(\Delta_S^{\rightarrow}  + \Delta_S^{\leftarrow}),
\label{eq_symmetric}
\end{eqnarray}
which  reaches its maximum value if  sets of states are LOCC distinguishable from both sides.

Before presenting the general result, let us consider an ensemble consisting of  any two maximally entangled states, given by
\begin{eqnarray}
E^{1}_{\max} = \{ |\phi{\pm}\rangle = \frac{1}{\sqrt{2}}(|0 0\rangle \pm  |11\rangle)\}. 
\end{eqnarray}
The local entropy of the average initial state is \(1\) while after CNOT operation, \( S(\overline{\rho}^{'A}) =1\) and \( S(\overline{\rho}^{'B}) =0\). Therefore, the maximum is \(1\) which is the maximal reduction in extractable work by the use of a known single qubit state. Note that if \(|\phi^-\rangle \) is replaced by \(|\psi^-\rangle\), one has to perform first single qubit unitary operation, say \(\sigma_z\) followed by CNOT operation which yields the same maximum value. In the case of $\{\phi^+,\phi^-,\psi^-\}$ and for a set having all the two-qubit maximally entangled  states, we obtain $\Delta_S^{\rightarrow}$  to be $0.0817$ and zero respectively.  Notice that the last two ensembles are LOCC-indistinguishable \cite{Ghosh'01} while \(E^1_{\max}\) can be distinguished by LOCC.

\emph{Result 1.} Let us now consider an ensemble consisting of two othergonal states \cite{WHSV}. We know that any two orthogonal states are always locally distinguishable \cite{WHSV} if \(A\) starts the protocol, and are given by 
\begin{equation}
E^2 = \{ |\psi_1\rangle = |0 \eta_1\rangle +  |1 \eta_2\rangle, \,\,  |\psi_2\rangle =  |0 \eta_1^{\perp}\rangle +  |1 \eta_2^{\perp} \rangle\},
\end{equation}
where \(|\eta_i \rangle = a_i |0\rangle + b_i |1\rangle, \, i=1, 2\) are arbitrary  single-qubit states and  \(|\eta_i^{\perp} \rangle = - \overline{b}_i |0\rangle + \overline{a}_i |1\rangle, \, i=1, 2\)  are the corresponding orthogonal ones. Without loss of generality, we assume that \(a_i \geq b_i\).  The initial  local entropies of the average state are \(S(\overline{\rho}^A)\) and \(S(\overline{\rho}^B)\).
Let us now apply CNOT operations on \(|\psi_i\rangle\)s where qubit at \(A\)'s side acts as a control qubit while the target qubit is the qubit at \(B\)'s side, and we get 
\begin{equation}
E^{'2} = \{|\psi_1\rangle = |0 \eta_1\rangle +  |1 \eta'_2\rangle; \,\, |\psi_2\rangle =  |0 \eta_1^{\perp}\rangle +  |1 \eta_2^{'\perp} \rangle\},
\end{equation}
where  \(|\eta'_2\rangle = a_2  |1\rangle + b_2 |0\rangle\) and  \(|\eta_2^{'\perp} \rangle = - \overline{b}_2 |1\rangle + \overline{a}_2 |0\rangle\). The corresponding final local entropies of the average state are  \(S(\overline{\rho^{\prime}}^A) \) and \(S(\overline{\rho^{\prime}}^B) \). Similarly, one can compute \(S(\overline{\rho^{\prime}}^A) \) and \(S(\overline{\rho^{\prime}}^B) \) when the role at \(A\) and \(B\) with respect to CNOT operation gets reversed. For example, when $a_1=\frac{4}{5}, a_2=\frac{3}{4}$, we obtain $\Delta_S^{\rightarrow}=0.0007$ while, $\Delta_S^{\leftarrow}=0$.
We find $\Delta_S^{\rightarrow}>\Delta_S^{\leftarrow}$ for all nonvanishing \(a_i\) and \(b_i\)s. Interestingly, we note that the set is distinguishable when left party starts the protocol but not the other way around.

\emph{ Result 2.}
In \(\mathcal{C}^{2} \otimes \mathcal{C}^{2}\),  a full basis containing non-maximally entangled states, given by
\begin{eqnarray}
E^3 =\{|\psi_1\rangle =a |00\rangle + b |11\rangle; \nonumber\\
|\psi_2 \rangle = - b |00\rangle +a |11\rangle; \nonumber\\
 |\psi_3 \rangle = a |01\rangle + b |10\rangle; \nonumber\\
 |\psi_4 \rangle = -b |01\rangle +a |10\rangle\},
\end{eqnarray}
with \(a\) and \(b\) being real are known to be  indistinguishable by LOCC \cite{Ghosh'01}. 
Following similar procedure as above, we find \(\Delta_S^{\rightarrow}\) is vanishing for the full set while unity for sets containing any two states from the set. Interestingly, if we consider first three states,  the value remains strictly in between as expected. For clarity,  we provide the expression for that case, $\Delta_S^{\rightarrow}=\frac{1}{3}[2-(2-b^2) \log_2 (2-b^2)-(1+b^2) \log_2 (1+b^2)]$. Similar result can also be obtained for complex \(a\) and \(b\).

\emph{Result 3.}
 Let us now discuss the ensemble which demonstrates more nonlocality with less entanglement containing states in  \(\mathcal{C}^{3} \otimes \mathcal{C}^{3}\). To explain it, we first consider a set of maximally entangled states, given by
  \cite{morenlwe}
\begin{eqnarray}
E_{3\otimes 3} = &&\{|\psi_1\rangle = \frac{1}{\sqrt{3}}(|00\rangle + \omega |11\rangle + \omega^2 |22\rangle); \nonumber \\
&&|\psi_2 \rangle = \frac{1}{\sqrt{3}}(|00\rangle + \omega^2 |11\rangle + \omega |22\rangle); \nonumber \\
&& |\psi_3 \rangle = \frac{1}{\sqrt{3}}(|01\rangle +  |12\rangle + |20\rangle) \},
\label{eq_morenonlocal}
\end{eqnarray}
whereas the second set contains first two states from the  set with last one being replaced by $|01\rangle$. It was shown that the first set is LOCC-distinguishable while the second is not \cite{morenlwe}.
It can be easily checked that by operating CNOT twice, we can disentangle all the states in \(E_{3\otimes 3}\) as well as the second set. 
It turns out that for the first set, \(\Delta_S^{\rightarrow}\) is  $\log_2 3$  
while  for the second set, it is less than \(\log_2 3\), specifically, $1.43552$, thereby revealing their difference with respect to LOCC-discrimination task. 
Apart from the LOCC protocol which differentiates these two sets, the above observation reveals the physical explanation of the phenomenon.

 It is to be mentioned here that if we increase the dimension of the systems, to maximize \(\Delta_S\), we have to operate CNOT  several times, preceded by a few local unitary operators, which overall constitute a global unitary operation, as shown before.

\subsection{Ensembles with maximally entangled states in two-qudits}

We end this section considering an extension of our result to a set of maximally entangled states in \(\mathcal{C}^{d} \otimes \mathcal{C}^{d}\).

\begin{theorem}
In two-qudits, any \(d\) orthogonal maximally entangled states written in a canonical form maximizes \(\Delta_S\) while for \(d^2\) states, it vanishes. \(\Delta_S\) for the ensemble consisting of \(d+1\) to \(d^2-1\) maximally entangled states chosen in that form lies in between maximum and minimum. 
\end{theorem}

\begin{proof}
In  \(\mathcal{C}^{d} \otimes \mathcal{C}^{d}\), any full orthogonal maximally entangled basis in a canonical form shared by \(A\) and \(B\) can be written as 
\begin{eqnarray}
\label{eq_cano}
|\Psi_{lm}\rangle = \frac{1}{\sqrt{d}} \sum_{k=0}^{d-1} \exp(\frac{2 \pi i l k}{d}) |k\rangle \otimes |k \oplus_d m\rangle,
\end{eqnarray}
where \(l, m = 0,  1, \ldots d-1\). For a given \(m\) values, there are \(d\) pairwise orthogonal maximally entangled states, lying in the same subspace and  we refer them as a block. Hence, the  full bases of maximally entangled state of the form, given above, consists of \(d\) such blocks. By applying CNOT operation in Eq. (\ref{eq_cano}) from \(A\)  to \(B\) \(d-1\) times, it can be found that the reduced state at \(B\)'s side obtained from each state in the \(k\)-th block becomes  \(|k\rangle,\) where, \( k=0,1,\ldots d-1\), which implies that the average reduced state is a pure state. 
 Hence if \(d\) number of states are chosen from the same block for the purpose of LOCC-discrimination task, then \(S(\overline{\rho^{\prime}}^{B}) =0\), thereby maximizing \(\Delta_S = \log_2 d\). On the other hand, if \(d\) states are chosen from different blocks, we have to transfer them to a same block by local unitary transformations first and then applies CNOT operation which again maximizes \(\Delta_S\). 

If  the full basis i.e. \(d^2\) states are given,  \(\rho'^{A}\) as well as \(\rho'^{B}\) are maximally mixed after CNOT operation which is same as the initial reduced density matrix of the average state, thereby \(\Delta_S\) vanishes for this ensemble. 

Let us now consider an ensemble consisting of  \(d+k\) orthogonal maximally entangled states.    By local unitary  operations,  \(d\) states can be transferred to a same block having the same value of \(m\) while the rest \(k\) states remain in other blocks. After  applying CNOT  to all the states in the ensemble  \(d-1\) times,  \(\rho'^{B(A)}\) is neither a pure state nor a maximally mixed state. However, the initial average local entropy in this situation still remains \(\log_2 d\) and hence \(\Delta_S\) for these ensembles lies between maximum and minimum.

\end{proof}

\section{Connection between non-local entropy with locally accessible information}

In quantum information theory, Holevo bound \cite{Holevo}, denoted by \(\chi\), is an important quantity which provides an upper bound on accessible information,  \(I_{acc}^{gl}\),  obtained by global measurements on an ensemble \(\{p_i, \rho_i\}\). Mathematically,  \(I_{acc}^{gl} \leq \chi \equiv S(\overline{\rho}) - \sum_i p_i S(\rho_i)\). Here \(\overline{\rho} = \sum_i p_i \rho_i\) and \(\chi\) is known as the Holevo quantity.  
On the other hand, if the measurements are restricted to local operations and classical communication, it is obvious that the information gained from the ensembles consisting of bipartite states would be less than the Holevo quantity. Evaluating accessible information exactly in q general scenario is a very hard task. However,
it was found that locally accessible information \(I_{acc}^{LOCC}\),  from an ensemble  \(\{p_i, \rho_i^{AB}\}\), is upper bounded by
\(  S(\overline{\rho}^A) + S(\overline{\rho}^B) - \max_{A,B}\sum_i p_i S(\rho_i^{A/B}) \) 
where \(\overline{\rho}^x\) is the local density matrix of the average ensemble state \cite{locallyaccess}.

 In the case of ensemble of product states, \(S(\rho_i^{A/B}) =0\) and hence the above mentioned upper bound does not provide a better bound for LOCC-indistinguishable product bases. Moreover, note that for LOCC-distinguishable sets of product states, locally accessible information coincides with that of accessible information by global means. Let us see whether our proposed quantifier can provide some bound on accessible information in a meaningful way.
 In the case of full product bases, if they are indistinguishable, we have shown that after the application of CNOT on the ensemble, FPB transforms to an ensemble which contains atleast a single entangled state. In case of full basis, it was shown \cite{morenlwe} that the sets are LOCC-indistinguishable if one (two) of the state is entangled. And hence, after CNOT operation, the resulting ensemble is surely LOCC indistinguishable and hence we can obtain a lower bound on locally accessible information for these ensemble, i.e., 
\(I_{acc}^{LOCC} ( \{p_i, |\psi_i^A \psi_i^B\rangle \}) \geq I_{acc}^{LOCC} ( \{p_i, U|\psi_i^A \psi_i^B\rangle \}) \).

 If the ensembles consist of entangled states, the opposite happens, and we get an upper bound, i.e., \(I_{acc}^{LOCC} ( \{p_i, |\psi^{AB}\rangle \}) \leq I_{acc}^{LOCC} ( \{p_i, U|\psi^{AB}\rangle \}) \leq  S(\overline{\rho^{\prime}}^A) + S(\overline{\rho^{\prime}}^B) -\max_{A,B} \sum_i p_i S(\rho_i^{\prime A/B})\), where the notations have usual meaning. Note that the second bound is effective when the set  retains some of its entangled states even after application of CNOT operation (eg. the set $E^2$ in previous subsection).

\section{Conclusion}
\label{sec_conclu} 

In any physical theory, how objects change under certain transformations, has utmost significance in benchmarking the theory. We apply this idea to connect two notions of nonlocality in quantum theory. In the first case, there exist certain correlations emerging from incompatible measurements on entangled systems, which can not be described by local realistic models --known as Bell-nonlocality. On the other hand, there are product states which are locally immeasurable, which is known as nonlocality without entanglement. The second type of nonlocality refers to the property of an ensemble while the first kind is for a particular composite state. Here we applied a global unitary transformation on the set of full product bases and showed that the resulting states behave  according to their  discrimination  protocol under local operations and classical communication (LOCC). In particular, we proposed a physical quantity, the generation of average entanglement from the ensemble under the considered transformation. The figure of merit, dubbed here as nonlocal entropy, tuned out to quantify faithfully the amount of nonlocality associated with the set. It vanishes for LOCC-distinguishable sets and becomes maximum for LOCC-indistinguishable sets of states which are indistinguishable irrespective of the parties, starting the protocol and even with the help of infinite rounds of classical communication.
 We also characterised successfully incomplete product bases including unextendible product bases with respect to the proposed quantifier.

We also showed that a modified quantifier under global unitary transformation can reveal the local (in)distinguishability properties of an ensemble consisting of at least one entangled state. 
 We demonstrated it by considering a set of nonmaximally entangled states in two-qubits as well as maximally entangled states in a canonical form in arbitrary dimension. Interestingly, 
 a satisfactory explanation for the phenomena of more nonlocality with less entanglement can be revealed via this quantification.

\section*{Acknowledgement}
We acknowledge the support from Interdisciplinary Cyber Physical Systems (ICPS) program of the Department of Science and Technology (DST), India, Grant No.: DST/ICPS/QuST/Theme- 1/2019/23.


\begin{thebibliography}{100}


\bibitem{speakablebook} J. S. Bell, \emph{Speakable and Unspeakable in Quantum Mechanics}, (Cambridge University Press,  New York, 1987). 

\bibitem{HHHHRMP} R. Horodecki, P. Horodecki, M. Horodecki, and K. Horodecki, Rev. Mod. Phys. {\bf 81}, 865 (2009)

\bibitem{epr} A. Einstein, B. Podolsky, and N. Rosen, Phys. Rev. {\bf 47}, 777 (1935).


\bibitem{bohm} D. Bohm,  \emph{Quantum Theory}, (New York: Prentice Hall, 1951); D. Bohm, and Y. Aharonov, Phys. Rev. {\bf 108},  1070 (1957).

\bibitem{Bell}  J. S. Bell, Physics {\bf 1}, 195 (1964).

\bibitem{Bellexpt} A. Aspect, J. Dalibard, and G. Roger,  Phys. Rev. Lett. {\bf 49}, 1804 (1982). 

\bibitem{Barrett'05} J. Barrett, L. Hardy, and A. Kent, Phys. Rev. Lett. {\bf 95}, 010503 (2005).

\bibitem{Acin'06} A. Ac{\' i}n, N. Gisin, and L. Masanes, Phys. Rev. Lett. {\bf 97}, 120405 (2006).

\bibitem{Acin'07} A. Ac{\' i}n, N. Brunner, N. Gisin, S.  Massar, S. Pironio, and V.  Scarani,  Phys. Rev. Lett. {\bf 98}, 230501 (2007).


\bibitem{Bellrev}N. Brunner, D. Cavalcanti, S. Pironio, V. Scarani, and S. Wehner, Rev. Mod. Phys. {\bf 86}, 419  (2014).




\bibitem{nlwe}C.H. Bennett, D.P. DiVincenzo, C.A. Fuchs, T. Mor, E.M. Rains, P.W. Shor, J.A. Smolin, and W.K. Wootters, Phys. Rev. A {\bf 59}, 1070 (1999).

\bibitem{PeresWootters} A. Peres, and W.K. Wootters, Phys. Rev. Lett. {\bf 66}, 1119 (1991). 

\bibitem{ChenLi04}P.-X. Chen, and C.-Z. Li, Phys. Rev.  A {\bf 70}, 022306 (2004).



\bibitem{WHSV}  J. Walgate, A.J. Short, L. Hardy, and V. Vedral, Phys. Rev. Lett. {\bf 85}, 4972
(2000).

\bibitem{NisetCerf} J. Niset and N. J. Cerf, Phys. Rev. {\bf 74}, 052103 (2006).


\bibitem{Werner} R. F. Werner, Phys. Rev A {\bf 40}, 4277 (1989).

 \bibitem{Spekkens'07} R. W. Spekkens, Phys. Rev. A {\bf 75}, 032110 (2007). 
 
 \bibitem{UPB1} C.H. Bennett, D.P. DiVincenzo, T. Mor, P.W. Shor, J.A. Smolin, and B.M. Terhal, Phys. Rev. Lett. {\bf 82}, 5385 (1999).

\bibitem{UPB2} D.P. DiVincenzo, T. Mor, P.W. Shor, J.A. Smolin, and B.M. Terhal, Unextendible Product Bases, Uncompletable Product Bases and Bound Entanglement, quant- ph/9908070.



\bibitem{groisman} B. Groisman, and L. Vaidman, J. Phys. A {\bf 34}, 6881 (2001).


 \bibitem{Rinaldis} S. De Rinaldis, Phys. Rev. A {\bf 70}, 022309 (2004).
 
 \bibitem{Nathanson} Michael Nathanson, J. of Math. Phys.  {\bf 46}, 062103 (2005).
 
\bibitem{others} Y.-L. Wang, M.-S. Li, Z.-J. Zheng, and S.-M. Fei, Phys. Rev. A {\bf 92}, 032313 (2015); Z.-C. Zhang, F. Gao, Y. Cao, S.-J. Qin, and Q.-Y. Wen, Phys. Rev.  A 93, 012314 (2016); S. Haldar, Phys. Rev. A {\bf 98}, 022303 (2018);  Mao-Sheng Li, Yan-Ling Wang,  Phys. Rev. A 98 052352(2018); S. Halder, M. Banik, S. Agrawal, and S. Bandyopadhyay; Phys. Rev. Lett. {\bf 122}, 040403 (2019).




\bibitem{Ghosh'01} S. Ghosh, G. Kar, A Roy, A. Sen(De) and U. Sen, Phys. Rev. Lett. {\bf 87 }, 277902  (2001); S. Ghosh, G. Kar, A Roy, D. Sarkar,  A. Sen(De) and U. Sen, Phys. Rev. A. {\bf 65}, 062307  (2002). 

\bibitem{WalgateHardy} J. Walgate and L. Hardy, Phys. Rev. Lett. {\bf 89}, 147901 (2002).

\bibitem{morenlwe} M. Horodecki, A. Sen(De), U. Sen, and K. Horodecki, Phys. Rev. Lett. {\bf 90}, 047902 (2003).

\bibitem{Somdamnlp} S. Bandyopadhyay, Phys. Rev. Lett. {\bf 106}, 210402 (2011). 

\bibitem{maxentdxd}  S. Ghosh, G. Kar, A Roy, D. Sarkar, Phys. Rev. A {\bf 70}, 022304 (2004); N. Yu, R. Duan, and M. Ying,
Phys. Rev. Lett. {\bf 109}, 020506 (2012); M. Nathanson, Phys. Rev.  A {\bf 88}, 062316 (2013); M.-S. Li,1 Y.-L. Wang, S.-M. Fei, and Z.-J. Zheng, Phys. Rev. A {\bf 91}, 042318 (2015); Y.-L. Wang, M.-S. Li, and Z.-X. Xiong, Phys. Rev. A {\bf 99},  022307 (2019); S. Bandyopadhyay, S. Ghosh, and G. Kar, New J. Phys. {\bf 13}, 123013 (2011); M.-S. Li, Y.-L. Wang, S.-M. Fei, and Z.-J. Zheng, Phys. Rev. A { \bf 91}, 042318 (2015); G.-J. Tian, X. Wu, Y. Cao, F. Gao, and Q.-Y. Wen, Sci. Rep. {\bf 6}, 30181 (2016).

\bibitem{datahide} B.M. Terhal, D.W. Leung, and D.P. DiVincenzo, Phys. Rev. Lett. {\bf 86}, 5807 (2001); D.P. DiVincenzo, D.W. Leung, and B.M. Terhal, IEEE Trans. Inf. Theory {\bf 48}, 580 (2002).

\bibitem{DCamader} D. Bru{\ss}, G. M. D’Ariano, M. Lewenstein, C. Macchiavello, A. Sen(De), and U. Sen, Phys. Rev. Lett. {\bf 93}, 210501 (2004); 
 D. Bru{\ss}, G. M. D’Ariano, M. Lewenstein, C. Macchiavello, A. Sen(De) and U. Sen, Int. J. Quantum Inform. {\bf 04}, 415 (2006); 
 
 \bibitem{DCnoise} Z. Shadman, H. Kampermann, C. Macchiavello, and D. Bru{\ss} , New J. Phys. {\bf 12}, 073042 (2010); Phys. Rev. A {\bf 84}, 042309 (2011); {\bf 85}, 052306 (2012);T. Das, R. Prabhu, A. Sen(De), and U. Sen, Phys. Rev. A {\bf 92}, 052330 (2015).

\bibitem{locallyaccess1}P. Badziag, M. Horodecki, A. Sen (De), and U. Sen,  Phys. Rev. Lett. {\bf 91}, 117901 (2003).

\bibitem{locallyaccess}  M. Horodecki, J. Oppenheim, A. Sen (De), and U. Sen,  Phys. Rev. Lett. {\bf 93}, 170503 (2004); S. Ghosh, P. Joag, G. Kar, S. Kunkri, and A. Roy, Phys. Rev. A {\bf 71}, 012321 (2005).



\bibitem{locallyaccesslow} A. Sen (De), U. Sen, and M. Lewenstein, Phys. Rev. A {\bf 74}, 052332 (2006).

\bibitem{measureQCset} M. Horodecki, A. Sen (De), and U. Sen,  Phys. Rev. A {\bf 75}, 062329 (2007).

\bibitem{Holevo}  J.P. Gordon, in Proc. Int. School Phys. “Enrico Fermi, Course XXXI”, ed. P.A. Miles, pp 156 (Academic Press, NY 1964); L.B. Levitin, in Proc. VI National Conf. Inf. Theory, Tashkent, pp 111 (1969); A.S. Holevo, Probl. Pereda. Inf. {\bf 9}, 3 (1973) [Probl. Inf. Transm. 9, 110 (1973)].

\bibitem{Josza94}  R. Josza, D. Robb, and W. Wootters, Phys. Rev. A, {\bf 49}, 668 (1994).

\bibitem{Bennett96} C.H. Bennett, H. Bernstein, S. Popescu, and B. Schumacher, Phys. Rev. A {\bf 53}, 2046 (1996).

\bibitem{gisin'91} N. Gisin, Phys. Lett. A {\bf 154}, 201 (1991).
\bibitem{wd} J. Oppenheim, M. Horodecki, P. Horodecki, R. Horodecki, Phys. Rev. Lett. {\bf 89}, 180402 (2002); M. Horodecki, P. Horodecki, R. Horodecki, J. Oppenheim, A Sen(De), U. Sen, and B. S. Radtke, Phys. rev. A {\bf 71}, 062307 (2005).





\bibitem{Barenco95}A. Barenco,  C. H. Bennett, R. Cleve, D. P. DiVincenzo, N. Margolus, P. Shor, T. Sleator, J. A. Smolin, and H. Weinfurter, Phys. Rev. A \textbf{52}, 3457 (1995). 
\bibitem{Zanardi00} P. Zanardi, C. Zalka, and L. Faoro, Phys. Rev. A 62,
030301(R) (2000).

\bibitem{FV04} F. Vatan and C. Williams, Phys. Rev. A {\bf 69}, 032315 (2004).


\bibitem{Chitambar} E. Chitambar, W. Cui, and H.-K. Lo,  Phys. Rev. Lett. {\bf 108}, 240504 (2012); Phys. Rev. A {\bf 85}, 062316 (2012).

\bibitem{aentset} Y. Cai, B. Yu, P. Jayachandran, N. Brunner, V. Scarani, J. D. Bancal, arXiv: 2006.07165.

\end{thebibliography}
\end{document}